\documentclass[pra,aps,showpacs,twoside,twocolumn,superscriptaddress]{revtex4}

\usepackage{amsmath,amsfonts,amssymb,caption,hyperref,color,epsfig,graphics,graphicx,latexsym,mathrsfs,revsymb,theorem,url,verbatim,epstopdf}

\hypersetup{colorlinks,linkcolor={blue},citecolor={blue},urlcolor={red}} 

\usepackage{hyperref}

\theoremstyle{plain}
\newtheorem{definition}{Definition}
\newtheorem{proposition}[definition]{Proposition}
\newtheorem{lemma}{Lemma}

\newtheorem{theorem}{Theorem}
\newtheorem{corollary}{Corollary}
\newtheorem{conjecture}{Conjecture}

\newtheorem{remark}{Remark}
\newtheorem{example}{Example}
\newtheorem{question}[definition]{Question}

\def\squareforqed{\hbox{\rlap{$\sqcap$}$\sqcup$}}
\def\qed{\ifmmode\squareforqed\else{\unskip\nobreak\hfil
\penalty50\hskip1em\null\nobreak\hfil\squareforqed
\parfillskip=0pt\finalhyphendemerits=0\endgraf}\fi}
\def\endenv{\ifmmode\;\else{\unskip\nobreak\hfil
\penalty50\hskip1em\null\nobreak\hfil\;
\parfillskip=0pt\finalhyphendemerits=0\endgraf}\fi}
\newenvironment{proof}{\noindent \textbf{{Proof.~} }}{\qed}
\def\Dbar{\leavevmode\lower.6ex\hbox to 0pt
{\hskip-.23ex\accent"16\hss}D}
\makeatletter
\def\url@leostyle{%
  \@ifundefined{selectfont}{\def\UrlFont{\sf}}{\def\UrlFont{\small\ttfamily}}}
\makeatother
\def\bcj{\begin{conjecture}}
\def\ecj{\end{conjecture}}
\def\bcr{\begin{corollary}}
\def\ecr{\end{corollary}}
\def\bd{\begin{definition}}
\def\ed{\end{definition}}
\def\bea{\begin{eqnarray}}
\def\eea{\end{eqnarray}}
\def\bem{\begin{enumerate}}
\def\eem{\end{enumerate}}
\def\bex{\begin{example}}
\def\eex{\end{example}}
\def\bim{\begin{itemize}}
\def\eim{\end{itemize}}
\def\bl{\begin{lemma}}
\def\el{\end{lemma}}
\def\bma{\begin{bmatrix}}
\def\ema{\end{bmatrix}}
\def\bpf{\begin{proof}}
\def\epf{\end{proof}}
\def\bpp{\begin{proposition}}
\def\epp{\end{proposition}}
\def\bqu{\begin{question}}
\def\equ{\end{question}}
\def\br{\begin{remark}}
\def\er{\end{remark}}
\def\bt{\begin{theorem}}
\def\et{\end{theorem}}

\def\btb{\begin{tabular}}
\def\etb{\end{tabular}}

\newcommand{\nc}{\newcommand}

\def\a{\alpha}
\def\b{\beta}

\def\d{\delta}

\def\r{\rho}
\def\s{\sigma}

 \nc{\bbA}{\mathbb{A}} \nc{\bbB}{\mathbb{B}} \nc{\bbC}{\mathbb{C}}
 \nc{\bbD}{\mathbb{D}} \nc{\bbE}{\mathbb{E}} \nc{\bbF}{\mathbb{F}}
 \nc{\bbG}{\mathbb{G}} \nc{\bbH}{\mathbb{H}} \nc{\bbI}{\mathbb{I}}
 \nc{\bbJ}{\mathbb{J}} \nc{\bbK}{\mathbb{K}} \nc{\bbL}{\mathbb{L}}
 \nc{\bbM}{\mathbb{M}} \nc{\bbN}{\mathbb{N}} \nc{\bbO}{\mathbb{O}}
 \nc{\bbP}{\mathbb{P}} \nc{\bbQ}{\mathbb{Q}} \nc{\bbR}{\mathbb{R}}
 \nc{\bbS}{\mathbb{S}} \nc{\bbT}{\mathbb{T}} \nc{\bbU}{\mathbb{U}}
 \nc{\bbV}{\mathbb{V}} \nc{\bbW}{\mathbb{W}} \nc{\bbX}{\mathbb{X}}
 \nc{\bbZ}{\mathbb{Z}}

 \nc{\bA}{{\bf A}} \nc{\bB}{{\bf B}} \nc{\bC}{{\bf C}}
 \nc{\bD}{{\bf D}} \nc{\bE}{{\bf E}} \nc{\bF}{{\bf F}}
 \nc{\bG}{{\bf G}} \nc{\bH}{{\bf H}} \nc{\bI}{{\bf I}}
 \nc{\bJ}{{\bf J}} \nc{\bK}{{\bf K}} \nc{\bL}{{\bf L}}
 \nc{\bM}{{\bf M}} \nc{\bN}{{\bf N}} \nc{\bO}{{\bf O}}
 \nc{\bP}{{\bf P}} \nc{\bQ}{{\bf Q}} \nc{\bR}{{\bf R}}
 \nc{\bS}{{\bf S}} \nc{\bT}{{\bf T}} \nc{\bU}{{\bf U}}
 \nc{\bV}{{\bf V}} \nc{\bW}{{\bf W}} \nc{\bX}{{\bf X}}
 \nc{\bZ}{{\bf Z}}

\nc{\cA}{{\cal A}} \nc{\cB}{{\cal B}} \nc{\cC}{{\cal C}}
\nc{\cD}{{\cal D}} \nc{\cE}{{\cal E}} \nc{\cF}{{\cal F}}
\nc{\cG}{{\cal G}} \nc{\cH}{{\cal H}} \nc{\cI}{{\cal I}}
\nc{\cJ}{{\cal J}} \nc{\cK}{{\cal K}} \nc{\cL}{{\cal L}}
\nc{\cM}{{\cal M}} \nc{\cN}{{\cal N}} \nc{\cO}{{\cal O}}
\nc{\cP}{{\cal P}} \nc{\cQ}{{\cal Q}} \nc{\cR}{{\cal R}}
\nc{\cS}{{\cal S}} \nc{\cT}{{\cal T}} \nc{\cU}{{\cal U}}
\nc{\cV}{{\cal V}} \nc{\cW}{{\cal W}} \nc{\cX}{{\cal X}}
\nc{\cZ}{{\cal Z}}

\nc{\hA}{{\hat{A}}} \nc{\hB}{{\hat{B}}} \nc{\hC}{{\hat{C}}}
\nc{\hD}{{\hat{D}}} \nc{\hE}{{\hat{E}}} \nc{\hF}{{\hat{F}}}
\nc{\hG}{{\hat{G}}} \nc{\hH}{{\hat{H}}} \nc{\hI}{{\hat{I}}}
\nc{\hJ}{{\hat{J}}} \nc{\hK}{{\hat{K}}} \nc{\hL}{{\hat{L}}}
\nc{\hM}{{\hat{M}}} \nc{\hN}{{\hat{N}}} \nc{\hO}{{\hat{O}}}
\nc{\hP}{{\hat{P}}} \nc{\hR}{{\hat{R}}} \nc{\hS}{{\hat{S}}}
\nc{\hT}{{\hat{T}}} \nc{\hU}{{\hat{U}}} \nc{\hV}{{\hat{V}}}
\nc{\hW}{{\hat{W}}} \nc{\hX}{{\hat{X}}} \nc{\hZ}{{\hat{Z}}}

\nc{\hn}{{\hat{n}}}

\makeatletter

\def\diag{\mathop{\rm diag}}

\def\min{\mathop{\rm min}}

\def\sep{\mathop{\rm SEP}}

\def\tr{\mathop{\rm Tr}}

\makeatletter

\newcommand{\Rmnum}[1]{\expandafter\@slowromancap\romannumeral #1@}
\makeatother

\newcommand{\bra}[1]{\langle#1|}
\newcommand{\ket}[1]{|#1\rangle}

\newcommand{\norm}[1]{\lVert#1\rVert}

\def\Dbar{\leavevmode\lower.6ex\hbox to 0pt
{\hskip-.23ex\accent"16\hss}D}

\begin{document}
	\title{Quantifying the entanglement of quantum states under the geometric method}

\author{Xian Shi}\email[]
{shixian01@gmail.com}
\affiliation{College of Information Science and Technology, Beijing University of Chemical Technology, Beijing 100029, China}
\affiliation{LMIB(Beihang University), Ministry of Education, and School of Mathematical Sciences, Beihang University, Beijing 100191, China}

\author{Lin Chen}\email[]{linchen@buaa.edu.cn (corresponding author)}
\affiliation{LMIB(Beihang University), Ministry of Education, and School of Mathematical Sciences, Beihang University, Beijing 100191, China}
\affiliation{International Research Institute for Multidisciplinary Science, Beihang University, Beijing 100191, China}
\author{Yixuan Liang}\email[]{linchen@buaa.edu.cn (corresponding author)}
\affiliation{LMIB(Beihang University), Ministry of Education, and School of Mathematical Sciences, Beihang University, Beijing 100191, China}
%



\date{\today}

\begin{abstract}
	Quantifying entanglement is an important issue in quantum information theory.   Here we consider the entanglement measures through the trace norm in terms of two methods, the modified measure and the extended measure for bipartite states.  We present the analytical formula for the pure states in terms of the modified measure and the mixed states of two-qubit systems for the  extended measure.  We also generalize the modified measure from bipartite states to tripartite states.
\end{abstract}
\maketitle
\section{Introduction}
\indent Quantum entanglement is an essential feature of quantum mechanics. It plays an important role in quantum information and quantum computation theory \cite{horodecki2009quantum}, such as superdense coding \cite{bennett1992communication}, teleportation \cite{bennett1993teleporting} and the speedup of quantum algorithms \cite{shimoni2005entangled}. \\
\indent How to quantify the amount of entanglement for a multipartite system is important in quantum information theory. As it
is linked with many areas of the entanglement theory, such as entanglement distillability \cite{bennett1996mixed}, the transformation of quantum states \cite{vedral1997quantifying,jonathan1999minimal,rains1999bound}, monogamy of entanglement \cite{coffman2000distributed,bai2014general,zhu2014entanglement,de2014monogamy,shi2021multilinear}, quantum speed limit \cite{rudnicki2021quantum} and so on. Due to the importance of this issue, it has been
 investigated almost since the end of the last century \cite{shimony1995degree,bennett1996mixed,vedral1997quantifying}. In 1996, Bennett $et$ $al.$ proposed the distillable entanglement and entanglement cost, as well as their operational interpretations \cite{bennett1996mixed}. In 1997, Vedral presented three necessary conditions that an entanglement measure should satisfy \cite{vedral1997quantifying}. And Vidal presented a general method to quantify the entanglement for the bipartite entangled states \cite{vidal2000entanglement}. This method is built on the function of pure states in the bipartite system and generalized to the mixed states under the convex roof extended methods on the bipartite system. Other than the generic conditions and the building methods of entanglement measures, some meaningful thoughts to bulid the entanglement measures are presented. The robustness of an entangled state tells us how many degrees of the separable states needed to make the state separable \cite{vidal1999robustness}. The squashed entanglement was motived by the intrinsic information \cite{tucci2002entanglement,christandl2004squashed}, it is monogamous for arbitrary dimensional systems \cite{christandl2004squashed} and faithful \cite{brandao2011faithful}.  Moreover, the way to measure the entanglement can also be defined based on the geometrical ways, such as, the geomeric measure of entanglement \cite{barnum2001monotones,wei2003geometric}, the quantum relative entropy \cite{vedral1997quantifying} and the fidelity \cite{streltsov2010linking}. However, there are few results of the measure based on the distance between an entangled state and the set of separable states under the trace norm. \par
 Trace norm not only can quantify the entanglement for a multipartitie mixed state, but also it provides a way to quantify the discord \cite{paula2013geometric,montealegre2013one,ciccarello2014toward,roga2016geometric}, the measurement-induced nonlocality \cite{hu2015measurement}, asymmetry \cite{marvian2014extending}, steering \cite{sainz2018formalism} and coherence \cite{baumgratz2014quantifying,rana2016trace, chen2018notes}. It is also helpful in quantum communication \cite{pirandola2017fundamental,pirandola2020advances}  and quantum algorithms \cite{gebhart2021quantifying,bai2022quantum}.
 \par
This article is organized as follows. In section \ref{se2}, first we present the preliminary knowledge needed. We also present some properties of the modified entanglement measure. In section \Rmnum{3}, we consider the properties of the entanglement mesurement in terms of the trace norm under two methods. Specifically, in section \ref{su2}, we present the solutions of the modified measure for pure states, and then we generalize it to a measure for tripartite mixed states $D^{'}_{fsep}(\cdot)$, which tells the distance between a tripartite state and the cone of full separable states.  In section \ref{su3}, we consider the measure based on the trace norm under the method proposed in \cite{gour2020optimal,shi2021extension}. In section \ref{se4}, we end with a conclusion.
\section{Preliminary Knowledge}\label{se2}
\indent In the following, we denote $\mathcal{D}_{A}$ as the set of states on $\mathcal{H}_{A}$ with finite dimensions. Next for a state $\rho_{AB}\in \mathcal{D}_{AB}$, if it can be written as $\rho_{AB}=\sum_i p_i\ket{\psi_i}_A\bra{\psi_i}\otimes\ket{\phi_i}_B\bra{\phi_i}$, then it is separable, otherwise, it is an entangled state. We denote $\sep_{AB}$ as the set of separable states on $\mathcal{H}_{AB}$ with finite dimensions.  And we may leave out the systems in the absence of ambiguity in the following. \par
Assume $\ket{\psi}_{AB}$ is a bipartite pure state in $\mathcal{H}_{AB}$, $\ket{\psi}_{AB}=\sum_i\sqrt{\lambda_i}\ket{ii}$, then its concurrence is defined as
\begin{align}
C(\ket{\psi}_{AB})=\sqrt{2(1-\tr\rho_A^2)},
\end{align}
here $\rho_A$ is the reduced density matrix of $\rho_{AB}.$ When $\rho_{AB}$ is a mixed state, then 
\begin{align}
C(\rho_{AB})=\min_{\{p_i,\ket{\psi}_{AB}\}}\sum_i p_iC(\ket{\psi_i}_{AB}),
\end{align}
where the minimum takes over all the decompositions $\{p_i,\ket{\psi_i}\}$ of $\rho_{AB}$. \par
Next we recall the definition of a distance measure in terms of trace norm for a bipartite mixed state. Assume $\rho\in \mathcal{D}_{AB}$, its distance measure in terms of trace norm is defined as,
	\begin{align}
	D_{sep}(\rho)=\min_{\b\in \sep}\norm{\rho-\b}_1, \label{e0}
	\end{align}
here $\b$ takes over all the separable states in the system $\mathcal{H}_A\otimes \mathcal{H}_B.$ \par
Then we recall $E:\mathcal{D}(\mathcal{H}_{AB})\rightarrow\mathbb{R}^{+}$ is an entanglement measure \cite{vedral1997quantifying} if it satifies:
\begin{itemize}
	\item [(\MakeUppercase{\romannumeral1})] $E(\rho_{AB})=0$ if $\rho_{AB}\in \mathcal{S}(\mathcal{H}_{AB})$.
	\item [(\MakeUppercase{\romannumeral2})]$E$ does not increase under the local operation and classical communication (LOCC),
	\begin{align*}
	E(\Psi(\rho_{AB}))\le E(\rho_{AB}),
	\end{align*} 
	here $\Psi$ is LOCC. 
\end{itemize}\par
In \cite{vidal2000entanglement}, the author presented that when $E$ satisfies the following two conditions, $E$ is an entanglement monotone,
\begin{itemize}
	\item [(\MakeUppercase{\romannumeral3})]$E(\rho)\ge \sum_k p_kE(\sigma_k),$ here $\sigma_{k}=\frac{\mathcal{E}_{i,k}(\rho_{AB})}{p_k},$ $p_k=\tr\mathcal{E}_{i,k}(\rho_{AB})$, $\mathcal{E}_{i,k}$ is any unilocal quantum operation performed by any party $A$ or $B,$ and $i$ are on behave of the party $A$ or $B.$
	\item [(\MakeUppercase{\romannumeral4})] For any decomposition $\{p_k,\rho_k\}$ of $\rho_{AB}$
	\begin{align*}
	E(\rho)\le \sum_k p_kE(\rho_k).
	\end{align*}
\end{itemize}
\par The condition (\MakeUppercase{\romannumeral4}) can also be regarded as the convexity of an entanglement measure. Next we recall the simplified conditions when the function satisifies the LOCC monotone \cite{horodecki2005simplifying}: \emph{For a convex function $f$ does not increase under LOCC if and only if
	\begin{itemize}
		\item[1]. $f$ satisfies local unitary invariance (LUI)
		\begin{align}
		f(U_A\otimes U_B\rho_{AB} U_A^{\dagger}\otimes U_B^{\dagger})=f(\rho_{AB}) \label{w0}
		\end{align}
		\item[2]. $f$ satisfies 
		\begin{align}
		f(\sum_i p_i \rho_{AB}^i\otimes\ket{i}_X\bra{i})=\sum_i p_i f(\rho_{AB}^i), \label{w1}
		\end{align}
		for $X=A^{'}, B^{'}$, where $\ket{i}$ are local orthogonal flags.
\end{itemize}}
\par

The properties (\MakeUppercase{\romannumeral1}), (\MakeUppercase{\romannumeral2}) and (\MakeUppercase{\romannumeral4}) are satisfied by $D_{sep}(\cdot)$ \cite{vedral1997quantifying,chen2014comparison}, however, the property (\MakeUppercase{\romannumeral3}) for the entanglement measure $D_{sep}(\cdot)$ is not valid \cite{qiao2017activation}. 
 Then we recall the modified version $D_{sep}^{'}(\cdot)$ of the entanglement measure $D_{sep}(\cdot)$, 
 \begin{align}
 D^{'}_{sep}(\rho_{AB})=\min_{\s\in SEP,\lambda>0}\norm{\rho-\lambda\s}_1.\label{mem}
 \end{align}
 \par Due to the properties of $\norm{\cdot}_1$ and the definition of $D_{sep}^{'}(\cdot),$ if we could prove 
\begin{align}
D_{sep}^{'}(\rho)\ge \sum_k p_kD^{'}_{sep}(\rho_{k}),\label{m0}\\
p_k=\tr\mathcal{E}_k(\rho_k),\hspace{3mm} \rho_k=\mathcal{E}_k(\rho_k)/p_k,\nonumber
\end{align}
here $\{\mathcal{E}_k\}$ is a quantum operation on the party $A$ at some stage ($k$ labels different outcomes if at some stage A performs a measurement), $D^{'}_{sep}(\cdot)$ satisfies the property (\MakeUppercase{\romannumeral3}). Next we recall the Naimark theorem \cite{paulsen2002completely}:\par
\emph{Assume $\{F_i\}_{i=1}^n$ is a POVM acting on a Hilbert space $\mathcal{H}_A$ with dimension $d_A$, then there exists a projection-valued measures (PVM) $\{\Pi_i\}_{i=1}^n$ acting on a Hilbert space $\mathcal{H}_{A^{'}}$ with dimension $d_{A^{'}}$, such that
	\begin{align}
	F_i=V^{\dagger}\Pi_i V,
	\end{align}
	A method is via the direct sum, requiring 
	\begin{align}
	\tr[F_i\rho]=\tr[P_i(\rho\oplus 0)],
	\end{align}
	here $\rho$ is an arbitrary operator in $\mathcal{H}_A$, $0$ is the zero matrix with dimension $d_A^{'}-d_A$.
} As the local operations $\mathcal{E}$, \emph{addition of an uncorrelated ancilla system and the dismissal of a local part of the whole system} satisfy
\begin{align*}
D^{'}_{sep}(\rho)\ge D^{'}_{sep}(\mathcal{E}(\rho)),
\end{align*}
the above inequality is due to the definition of $D_{sep}^{'}(\cdot)$ and the monotonicity of the trace norm under $\mathcal{E}.$ Combing the Naimark theorem, we only need to prove (\ref{m0}) is valid for the PVM, that is, we need to prove that
\begin{align}
D_{sep}^{'}(\rho)\ge \sum_k p_kD^{'}_{sep}(\rho_{k}),\label{m1}
\end{align}
here $p_k=\tr\Pi_k\rho,$ $\rho_{k}=\frac{\Pi_k\rho\Pi_k}{p_k}$, $\{\Pi_k\}$ is a PVM on the $A$-party.
\begin{align}
D_{sep}^{'}(\rho)=&\norm{\rho-\sigma}_1\nonumber\\
\ge &\norm{\sum_k(\Pi_k\rho\Pi_k-\Pi_k\sigma\Pi_k)}_1\nonumber\\
=&\sum_k\norm{\Pi_k\rho\Pi_k-\Pi_k\sigma\Pi_k}_1\nonumber\\
=&\sum_kp_k\norm{\frac{\Pi_k\rho\Pi_k}{p_k}-\sigma_{k}}_1\nonumber\\
\ge&\sum_kp_kD^{'}_{sep}(\rho_{k}). \label{m2}
\end{align}
Here in the first equality, $\sigma$ is the optimal in terms of $D^{'}_{sep}(\cdot)$ for $\rho$. In the first inequality, $\{\Pi_k\}$ is a PVM on the $A$-party of $\rho_{AB}$, $\sum_k\Pi_k=I$, then due to the contractive under the trace-preserving quantum operations, the first inequality is valid. In the third equality, $p_k=\tr \Pi_k \rho\Pi_k,$ $\sigma_{k}=\frac{\Pi_k\sigma\Pi_k}{p_k},$ here the trace of $\sigma_k$ may be not 1, that is, $\sigma_k$ may be not a state. As $\s_k$ a separable state, and due to the definition of $D_{sep}(\cdot),$ the last inequality is valid. We finish the proof of the inequality $(\ref{m2})$.
\par
At last, we prove the convexity of $D_{sep}^{'}(\cdot)$.
Assume $\{p_k,\rho_k\}$ is an arbitrary decompostion of $\rho_{AB}$, $\sigma_{k}$ is the optimal for the state $\rho_{k}$ in terms of $D^{'}_{sep}(\cdot)$, then 
\begin{align}
\sum_k p_kD^{'}_{sep}(\rho_k)=&\sum_kp_k ||\rho_k-\sigma_{k}||_1\nonumber\\
\ge&\norm{\sum_k(p_k\rho_k-p_k\sigma_{k})}_1\nonumber\\
\ge&\norm{\rho-\sigma}_1=D^{'}_{sep}(\rho) ,
\end{align}
the first equality is due to the definition of $(\ref{mem}),$ the first inequality is due to the convexity of the $\norm{\cdot}_1$. In the second inequality, we denote $\sigma$ as the optimal in terms of $D_{sep}^{'}(\cdot)$ for $\rho$, and the second inequality is due to the definition of $D_{sep}^{'}(\cdot).$ Then we prove the property (\MakeUppercase{\romannumeral4}) is satisfied by the entanglement measure $D_{sep}(\cdot)$.\par
\section{Main Results}
\indent  In this section, we will consider the entanglement measures based on trace norm in terms of two methods, the modified distance measure and the extended distance measure. First  we present  the analytical formula of a class of entangled states in terms of the modified entanglement measure, we also present some properties of the modified measure. At last, we consider the extended measure generated under the method presented in \cite{gour2020optimal,shi2021extension} and show this measure is monogamous for $n$-qubit systems .

 \subsection{ The modified distance measure }\label{su2}
\indent In this section, we first rewrite the definition of modified distance measure $(\ref{mem})$ as follows,
 \begin{align}
 D_{sep}^{'}(\rho)=\min_{\sigma}\norm{\rho-\sigma}_{1},
 \end{align}
 here $\sigma$ takes over all the elements in the set $\{\sigma|\sigma=\sum_i q_i\omega_i^A\otimes\zeta_i^B, q_i\ge 0, \omega_i^A\in D(\mathcal{H}_A), \zeta_i^B\in D(\mathcal{H}_B)\}$. In \cite{yu2016alternative}, the authors presented the method to quantify the coherence in this way. In \cite{johnston2018modified}, they showed that the modified distance of almost all pure states is a constant, and they numerically showed that a similar result occurs for states with a fixed rank. Here we showed that the same result occurs for the entangled pure states.
 \begin{theorem}\label{t3}
 	Assume $\ket{\psi}=\sum_i\lambda_i\ket{ii}$ is an arbitrary pure state, here $\lambda_i\ge \lambda_{i+1},$then the modified distance measure of $\ket{\psi}$ is 
 	\begin{align}
 	D^{'}_{sep}(\ket{\psi})=\left\{
 	\begin{aligned}
 	1\hspace{7mm} & , & \lambda_1\le\frac{1}{\sqrt{2}}, \\
 	2\lambda_1\sqrt{1-\lambda_1^2} & , & \lambda_1>\frac{1}{\sqrt{2}}.
 	\end{aligned}
 	\right.
 	\end{align}
 \end{theorem}
 We present the proof of this theorem in Sec. \ref{App}.\par
\begin{remark}
	In \cite{regula2019one}, the authors considered the same problem for the pure states in terms of $D_{sep}^{'}(\cdot)$. The method there is due to the convex analysis, ours is different from there. Moreover, comparing with the result in \cite{regula2019one}, ours is more apparent. And we can also generalize our methods to a class of mixed states.
\end{remark}
\begin{corollary}
	Assume $\rho=a_0\ket{00}\bra{00}+a_1\ket{00}\bra{11}+\overline{a_1}\ket{11}\bra{00}+a_2\ket{11}\bra{11}$ is a two-qubit mixed state, then $D_{sep}^{'}(\rho)=2|a_{1}|$.
\end{corollary}\par
\indent The main proof of the above corollary is the same as the proof for the pure states above, and by combing the theorem 1 in \cite{johnston2018modified}, we can get the corollary. Next we present some properties of the modified distance measure. 
\begin{theorem}
	Assume $\rho_{AB}$ and $\sigma_{AB}$ are two mixed states with $\norm{\r_{AB}-\s_{AB}}\le\epsilon$, then we have 
	\begin{align}
	|D^{'}_{sep}(\rho_{AB})-D^{'}_{sep}(\s_{AB})|\le \epsilon
	\end{align}
\end{theorem}\par

\begin{proof}
	Here we can always assume $D^{'}_{sep}(\rho_{AB})-D^{'}_{sep}(\s_{AB})\ge0$, then we have
	\begin{align}
	&D^{'}_{sep}(\r_{AB})-D^{'}_{sep}(\s_{AB})\nonumber\\
	\le&\norm{\r_{AB}-\a_{AB}}_1-\norm{\s_{AB}-\a_{AB}}_1\nonumber\\
	\le&\norm{\rho_{AB}-\s_{AB}}_1\le \epsilon.
	\end{align}
	Here $\a_{AB}$ is the optimal for $\s_{AB}$. The first inequality is due to the definition of $D^{'}_{sep}(\cdot)$, as $\a_{AB}$ may not be the optimal for $\r_{AB}$, the second inequality is due to the triangle inequality
	\begin{align*}
	&\norm{\r_{AB}-\a_{AB}}_1-\norm{\s_{AB}-\a_{AB}}_1\nonumber\\
	\le&\norm{\r_{AB}-\a_{AB}-\s_{AB}+\a_{AB}}_1\nonumber\\
	=&\norm{\r_{AB}-\s_{AB}}_1
	\end{align*} 
\end{proof}

 At the end of this subsection, we generalize the measure $D^{'}_{sep}(\rho)$ to a quantity $D^{'}_{fsep}(\cdot)$ for a tripartite system which tells the distance from a tripartite state to the set of fully separable states. Next we recall the fully separable states. A pure state $\ket{\psi}_{ABC}$ is fully separable if it can be written as $\ket{\psi}_{ABC}=\ket{\phi_1}_A\ket{\phi_2}_B\ket{\phi_3}_C$.  If a mixed state $\rho_{ABC}$ can be written as
\begin{align*}
\rho_{ABC}=\sum_{i=1}^kp_i\rho_A^i\otimes\rho_B^i\otimes\rho_C^i,
\end{align*} then $\rho_{ABC}$ is a fully separable state. The definition of $D^{'}_{fsep}(\rho_{ABC})$ is defined as follows,
\begin{align}
D^{'}_{fsep}(\rho_{ABC})=\inf_{\gamma\in FSEP^{'}} \norm{\rho-\gamma}_1,
\end{align}
here the infimum takes over all the elements in $\{\gamma|\gamma=\sum_{i=1}^kp_i\rho_A^i\otimes\rho_B^i\otimes\rho_C^i, p_i\ge 0, \rho_A^{i}\in\mathcal{D}(\mathcal{H}_A),\rho_B^{i}\in\mathcal{D}(\mathcal{H}_B),\rho_C^{i}\in\mathcal{D}(\mathcal{H}_C),i=1,2,\cdots,k.\}$. 
\par
In the following, we will compute the values $D^{'}_{fsep}(\ket{W})$ and $D^{'}_{fsep}(\ket{GHZ}).$
\begin{example}\label{Ex1}
	Assume 
	\begin{align}
	\ket{W}=&\frac{1}{\sqrt{3}}(\ket{001}+\ket{010}+\ket{100}),\nonumber\\\ket{GHZ}=&\frac{1}{\sqrt{2}}(\ket{000}+\ket{111}).\nonumber
	\end{align} Then $D^{'}_{fsep}(W)=1, D^{'}_{fsep}(GHZ)=1.$
\end{example}
\par
We place the proof of Example \ref{Ex1} in Sec. \ref{App}.\\

\subsection{The extended distance measure}\label{su3}
\indent Here we consider another entanglement measure based on the trace norm under the method presented in \cite{gour2020optimal,shi2021extension}.  There the authors presented a way to generate the entanglement measure $\overline{E}$ for the mixed states from the entanglement measure $E$ on pure states.  Assume $\rho_{AB}\in B(\mathcal{H_A\otimes H_B})$, \begin{align}
\overline{E}(\rho_{AB})=\inf_{\ket{\psi}_{AB}\in \mathcal{R}(\rho_{AB})} E(\ket{\psi}_{AB}),\label{Ev} 
\end{align} 
here the infimum takes over all the pure states in the set 
\begin{align*}
\mathcal{R}(\rho_{AB})=\{ \psi_{AB}&\in \mathcal{H}_{AB}|\rho_{AB}=\Lambda(\psi_{AB}),\Lambda\in\mathcal{T},\\ &\textit{$\mathcal{T}$ is a class of quantum channels. }\},
\end{align*}
 Here we specifically consider the properties of the entanglement measures based on the trace norm under the method above,
\begin{align}
\overline{D}_{sep}(\rho_{AB})=&\inf_{\ket{\psi}_{AB}\in \mathcal{R}(\rho_{AB})} D_{sep}(\ket{\psi}_{AB}),\nonumber\\
\mathcal{R}(\rho_{AB})=&\{ \psi_{AB}\in \mathcal{H}_{AB}|\rho_{AB}=\Lambda(\psi_{AB}),\Lambda\in LOCC.\}. \label{m3}
\end{align}
\par Then we show the above quantity is an entanglement measure. When $\r_{AB}$ is a mixed state, $\ket{\psi}_{AB}$ is the optimal pure state in (\ref{m3}), assume $\Lambda$ is LOCC, then
\begin{widetext}
\begin{align*}
\overline{D}_{sep}(\rho_{AB})=&\inf_{\ket{\psi}_{AB}}\{D_{sep}(\ket{\psi}_{AB})|\rho_{AB}=\Lambda(\psi_{AB}),\Lambda\in LOCC.\}\nonumber\\
\ge&\inf_{\ket{\psi}_{AB}}\{D_{sep}(\ket{\psi}_{AB})|\mathcal{M}(\rho_{AB})=\mathcal{M}\circ\Lambda(\psi_{AB}),\Lambda\in LOCC.\}\\
\ge&\inf_{\ket{\psi}_{AB}}\{D_{sep}(\ket{\psi}_{AB})|\mathcal{M}(\rho_{AB})=\Theta(\psi_{AB}),\Theta\in LOCC.\}=\overline{D}_{sep}(\mathcal{M}(\rho_{AB})),
\end{align*}
\end{widetext}
here $\mathcal{M}(\cdot)$ is LOCC. Hence, the property (\MakeUppercase{\romannumeral2}) is valid. As a separable pure state can be transformed into a separable mixed state through LOCC, then the quantity $\overline{D}_{sep}(\cdot)$ of a separable state is 0. Hence, it is an entanglement measure.
\par 
 Next we consider the mixed states on two-qubit systems in terms of $\overline{D}_{sep}(\cdot).$

\begin{theorem}
	When $\rho_{AB}$ is a two-qubit mixed state, $\overline{D}_{sep}(\rho_{AB})$ is an entanglement monotone. Furthermore, 
	\begin{align*}
	\overline{D}_{sep}(\rho_{AB})=C(\rho_{AB})=\min_{\{p_i,\ket{\psi_i}_{AB}\}}\sum_i p_i\overline{D}_{sep}(\ket{\psi_i}_{AB}),
	\end{align*}
	where the minimum takes over all the decompositions $\{p_i,\ket{\psi_i}_{AB}\}$
of $\rho_{AB}.$ Furthermore, assume $\rho_{A|B_1B_2\cdots B_{n-1}}$ is an $n$-qubit mixed state, then 
\begin{align}
\overline{D}_{sep}^2(\rho_{A|B_1B_2\cdots B_{n-1}})\ge \sum_{i=1}^{n-1}\overline{D}_{sep}^2(\rho_{AB_i}).
\end{align}
\end{theorem}
\begin{proof}
	Assume $\rho_{AB}$ is a two-qubit mixed state, then
	\begin{align*}
	\overline{D}_{sep}(\rho_{AB})=&D_{sep}(\ket{\psi}_{AB})\nonumber\\
	=&C(\ket{\psi}_{AB})\nonumber\\
	=&\overline{C}(\rho_{AB})\\
	=&\min_{\{p_i,\ket{\psi_i}_{AB}\}}\sum_i p_i\overline{D}_{sep}(\ket{\psi_i}_{AB}),
	\end{align*} 
	where the minimum in the last equality takes over all the decompositions $\{p_i,\ket{\psi_i}_{AB}\}$
	of $\rho_{AB}.$  In the first equality, $\ket{\psi}_{AB}$ is the optimal for $\rho_{AB}$ in terms of $\overline{D}_{sep}(\cdot)$, and combing $D_{sep}(\ket{\psi}_{AB})=C(\ket{\psi}_{AB})$ \cite{chen2016quantifying}, we have the second equality is valid. As for any two-qubit pure state $\ket{\phi}_{AB},$ $C(\ket{\phi}_{AB})=D_{sep}(\ket{\phi}_{AB})$ \cite{chen2016quantifying}, and due to the definition of $\overline{D}_{sep}(\cdot)$ and $\overline{C}(\cdot)$, the third equality is valid. Next we show the last inequality,
	\begin{align*}
	\overline{D}^2_{sep}(\rho_{A|B_1B_2\cdots B_{n-1}})=&D^2_{sep}(\ket{\psi}_{A|B_1B_2\cdots B_{n-1}})\nonumber\\
	=&C^2(\ket{\psi}_{A|B_1B_2\cdots B_{n-1}})\nonumber\\
	\ge&C^2(\rho_{A|B_1B_2\cdots B_{n-1}})\nonumber\\
	\ge&\sum_{i=1}^{n-1}C^2(\rho_{AB_i})\nonumber\\
	=&\sum_{i=1}^{n-1}\overline{D}_{sep}^2(\rho_{AB_i}).
	\end{align*}
In	the first equality, we denote $\ket{\psi}$ as the optimal in terms of $\overline{D}_{sep}(\cdot).$ As $\ket{\psi}_{A|B_1B_2\cdots B_{n-1}}$ is a $2\otimes n$ pure state, then we have $D^2_{sep}(\psi_{A|B_1B_2\cdots B_{n-1}})=C^2(\psi_{A|B_1B_2\cdots B_{n-1}})$ \cite{chen2016quantifying}. The first inequality is due to that $\rho_{A|B_1B_2\cdots B_{n-1}}$ can be transformed by $\psi_{A|B_1B_2\cdots B_{n-1}}$ under LOCC. The last inequality is proved in \cite{osborne2006general}. As when $\rho_{AB}$ is a two-qubit mixed state, $C(\rho_{AB})=\overline{D}_{sep}(\rho_{AB})$, we have the last equality is valid.
\end{proof}

\par
\section{Conclusion}\label{se4}
\par In this paper, we have considered the entanglement measures in terms of the trace norm under two methods.  We have presented the modified distance measure satisfies convexity for arbitrary dimensional systems. And we also have presented the analytical expressions for some classes of states. Then we have shown the extended measure is an entanglement monotone for the 2-qubit systems and monogamous for $n$-qubit systems.  We think the above results may be not valid for larger systems, although we havenot found a counterexample. At last, we believe that our results are helpful in the study of monogamy relations for multipartite entanglement systems. And we hope our work could shed some light on related studies.

\section{Acknowledgement}
\indent  This work was supported by the NNSF of China (Grant No. 11871089), and the Fundamental Research Funds for the Central Universities (Grant Nos. KG12080401 and ZG216S1902).
\bibliographystyle{IEEEtran}
\bibliography{ref}
\section{Appendix}\label{App}
Here we present the proof of Example \ref{Ex1} and Theorem \ref{t3}.

\subsection{The proof of Example \ref{Ex1}}
$Example$ \ref{Ex1}:
	Assume 
	\begin{align}
	\ket{W}=\frac{1}{\sqrt{3}}(\ket{001}+\ket{010}+\ket{100}),\nonumber\\\ket{GHZ}=\frac{1}{\sqrt{2}}(\ket{000}+\ket{111}).\nonumber
	\end{align} Then $D_{fsep}(W)=1, D_{fsep}(GHZ)=1.$\\
\begin{proof}
	\indent As $\ket{W}$ satisfies the following property:
	\begin{align}
	U^{\otimes3}_{\theta}\ket{W}\bra{W}{U_{\theta}^{\dagger}}^{{\otimes3}}=\ket{W}\bra{W},
	\end{align}
	here $U_{\theta}=\diag(1,e^{i\theta})$, then 
	\begin{align}
	&\norm{\ket{W}\bra{W}-\s}_1\nonumber\\
	=&\norm{U^{\otimes3}_{\theta}(\ket{W}\bra{W}-\s){U_{\theta}^{\dagger}}^{\otimes3}}_1\nonumber\\
	=&\frac{1}{2\pi}\int_{0}^{2\pi}\norm{U^{\otimes3}_{\theta}(\ket{W}\bra{W}-\s){U_{\theta}^{\dagger}}^{\otimes3}}_1d\theta\nonumber\\
	\ge&\norm{\frac{1}{2\pi}\int_{0}^{2\pi}[\ket{W}\bra{W}-U^{\otimes3}_{\theta}\sigma {U^{\dagger}_{\theta}}^{\otimes3}]d\theta}_1\nonumber\\
	=&\norm{\ket{W}\bra{W}-\s_d}_1,\label{w}
	\end{align}
	here $\s_d$ can be written as
	\begin{align} &\s_d\nonumber\\=&m_0\ket{000}\bra{000}+m_1\ket{W}\bra{W}+m_2\ket{\overline{W}}\bra{\overline{W}}+m_3\ket{111}\bra{111},\nonumber
	\end{align}
	\begin{align*}
	\ket{\overline{W}}=&\frac{1}{\sqrt{3}}(\ket{011}+\ket{101}+\ket{110}),
	\end{align*}
	From (\ref{w}), we have  
	\begin{align*}
	D_{fsep}(\ket{W})=\min_{\s\in DD\cap FSEP}\norm{\ket{W}\bra{W}-\s}_1,
	\end{align*}
	here the set $DD=\{\d=m_0\ket{000}\bra{000}+m_1\ket{W}\bra{W}+m_2\ket{\overline{W}}\bra{\overline{W}}+m_3\ket{111}\bra{111}|m_0,m_1,m_2,m_3\ge 0.\}.$ Combing the Theorem 1 in \cite{yu2016separability}, the above can be written as the following,
	\begin{align}
	&D_{fsep}(\ket{W})=\inf_{\vec{m}\in M}|m_0|+|m_2|+|m_3|+|1-m_1|,\nonumber\\
	&M=\{(m_0,m_1,m_2,m_3)|m_0m_2\ge m_1^2/3,m_1m_3\ge m_2^2/3.\},
	\end{align}
	from computation, we have $D_{fsep}(\ket{W})=1.$\par
	Next we present the result on $\ket{GHZ}$. In \cite{eltschka2012entanglement}, the authors presented that $\ket{GHZ}$ satisfies the following properties: \begin{itemize}
		\item[i.] it is a symmetric state,
		\item[ii.]$\sigma_x^{\otimes3}\ket{GHZ}=\ket{GHZ},$
		\item[iii.]$U(\phi_1,\phi_2)\ket{GHZ}=\ket{GHZ},$\\ here $U(\phi_1,\phi_2)=e^{i\phi_1\sigma_z}\otimes e^{i\phi_2\sigma_z}\otimes e^{-i(\phi_1+\phi_2)\sigma_z},$ $\phi_1,\phi_2$ are arbitrary values.
	\end{itemize} As $\ket{GHZ}$ is a symmetric state, that is, 
	\begin{align*}
	V_{\s}\ket{GHZ}\bra{GHZ}=\ket{GHZ}\bra{GHZ}V_{\s^{'}}=\ket{GHZ}\bra{GHZ},
	\end{align*}
	here $\s,\s^{'}$ are arbitrary permutations $\s,\s^{'}\in \Omega_N,$ $\Omega_N$ isthe group of all permutations of an $N$-element set, and $V_{\s}\ket{\psi_1}\ket{\psi_2}\cdots\ket{\psi_{N}}=\ket{\psi_{\s(1)}}\ket{\psi_{\s(2)}}\cdots\ket{\psi_{\s(N)}}.$ Assume $\theta$ is the state attaining the minimum, then
	\begin{align*}
	&\norm{\ket{GHZ}\bra{GHZ}-\frac{\theta+V_{\s}\theta V_{\s^{'}}}{2}}_1\nonumber\\
	=&\norm{\frac{1}{2}(\ket{GHZ}\bra{GHZ}-\theta)+\frac{1}{2}(V_{\s}(\ket{GHZ}\bra{GHZ}-\theta)V_{\s^{'}})}_1\nonumber\\
	\le&\frac{1}{2}(\norm{\ket{GHZ}\bra{GHZ}-\theta}_1+\norm{V_{\s}(\ket{GHZ}\bra{GHZ}-\theta)V_{\s^{'}}}_1\nonumber\\
	=&\norm{\ket{GHZ}\bra{GHZ}-\theta}_1,
	\end{align*}
	here $\s,\s^{'}$ are arbitrary permutations. As $s$ and $s^{'}$ are arbitrary, $\theta$ is a symmetric state. By the similar method, $\theta$ should satisfy the same properties (i),(ii) and (iii). That is, the matrix of $\theta$ should be the following:
	$\begin{pmatrix}
	m_0&0&0&0&0&0&0&n\\
	0&m_1&0&0&0&0&0&0\\
	0&0&m_1&0&0&0&0&0\\
	0&0&0&m_1&0&0&0&0\\
	0&0&0&0&m_1&0&0&0\\
	0&0&0&0&0&m_1&0&0\\
	0&0&0&0&0&0&m_1&0\\
	n&0&0&0&0&0&0&m_0
	\end{pmatrix} $\par
	Then by simple computation, we have $D_{fsep}(GHZ)=1.$
\end{proof}
 \subsection{The proof of Theorem \ref{t3}}
$Theorem$ \ref{t3}:
	Assume $\ket{\psi}=\sum_i\lambda_i\ket{ii}$ is an arbitrary pure state, here $\lambda_i\ge \lambda_{i+1},$then the modifies distance measure of $\ket{\psi}$ is 
	\begin{align}
	D^{'}_{sep}(\ket{\psi})=\left\{
	\begin{aligned}
	1\hspace{7mm} & , & \lambda_1\le\frac{1}{\sqrt{2}}, \\
	2\lambda_1\sqrt{1-\lambda_1^2} & , & \lambda_1>\frac{1}{\sqrt{2}}.
	\end{aligned}
	\right.
	\end{align}
\begin{proof}
	Assume $\ket{\psi}$ is a pure state, $\ket{\psi}=\sum_i \sqrt{\lambda_i}\ket{ii}$,  then we have $(U\otimes \overline{U})\ket{\psi}\bra{\psi}(U\otimes \overline{U})^{\dagger}=\ket{\psi}\bra{\psi}$, here $U$ is a diagonal unitary matrix. Next assume $\s$ is an arbitrary positive operator which can be written as $\r=\sum_i \omega_i^A\otimes\theta_i^B$, here $\omega_i^A$ and $\theta_i^B$ are positive operators on the systems $A$ and $B$, respectively.
	
	\begin{align}
	\norm{\ket{\psi}\bra{\psi}-\r}_1
	= &\int_{U\in \mathbb{D}(U)}\norm{(U\otimes \overline{U})(\ket{\psi}\bra{\psi}-\r)(U\otimes\overline{U})^{\dagger}}_1dU\nonumber\\
	\ge&\norm{\int_{U\in \mathbb{D}(U)}(U\otimes \overline{U})(\ket{\psi}\bra{\psi}-\r)(U\otimes\overline{U})^{\dagger}dU}_1\nonumber\\
	=&\norm{\ket{\psi}\bra{\psi}-\int_{U\in \mathbb{D}(U)}(U\otimes\overline{U})\r(U\otimes U)^{\dagger}dU}_1\label{s1}
	\end{align}
	
	Then let $\sigma^{'}$ be the optimal in terms of the modified distance measure $D_{sep}^{'}(\cdot)$ for $\ket{\psi}$, we have 
	\begin{align}
	D_{sep}^{'}(\ket{\psi}_{AB})=&\norm{\ket{\psi}\bra{\psi}-\sigma^{'}}_1\nonumber\\
	\ge&\norm{\ket{\psi}\bra{\psi}-\int_{U\in \mathbb{D}(U)}(U\otimes\overline{U})\sigma^{'}(U\otimes \overline{U})^{\dagger}dU}_1,\label{s2}
	\end{align} 
	as $\sigma^{'}$ is the optimal, and $\int_{U\in \mathbb{D}(U)}(U\otimes\overline{U})\sigma^{'}(U\otimes\overline{U})^{\dagger}$ is an element in the set of $\lambda SEP$, then we have the inequality in $(\ref{s2})$ is an equality. Let 
	\begin{align*}
	\vartheta=&\int_{U\in \mathbb{D}(U)} (U\otimes \overline{U})\varrho(U\otimes \overline{U})^{\dagger} dU\nonumber\\
	(\vartheta)_{ijkl}=&\int\int \int \int z_i\overline{z_j}z_l\overline{z_k}\varrho_{ijkl} dz_i dz_jdz_kdz_l\nonumber\\
	=&\left\{\begin{aligned}
	\varrho_{ijkl},\hspace{3mm} \textit{if $(i,l)=(j,k)$,}\nonumber\\
	\varrho_{ijkl},\hspace{3mm} \textit{if $(i,l)=(k,j)$,}\nonumber\\
	0,\hspace{11mm}  otherwise.
	\end{aligned}\right. 
	\end{align*}
	here $\mathbb{D}(U)$ is the set consisting of all diagnoal unitary matrices, $\varrho_{ijkl}$ denotes the coefficient of the state $\varrho$ in the basis $\ket{ij}\bra{kl}$. Hence we only need to compute the minimum of the following equality, here the minimum takes over all the positive operators $\varrho$ that can be written as $\varrho=\sum_i \omega_i^A\otimes\theta_i^B,$ $\omega_i^A$ and $\theta_i^B$ are positive operators on the systems $A$ and $B$, respectively.
	\begin{widetext}
		\begin{align}
		&\norm{\ket{\psi}\bra{\psi}-\sum_{i\ne j}\varrho_{iijj}\ket{ii}\bra{jj}-\sum_{i,j}\varrho_{ijij}\ket{ij}\bra{ij}}_1\nonumber\\
		=&\norm{\sum_{i,j}\sqrt{\lambda_i\lambda_j}\ket{ii}\bra{jj}-\sum_{i\ne j}\varrho_{iijj}\ket{ii}\bra{jj}-\sum_{i}\varrho_{iiii}\ket{ii}\bra{ii}}_1+\norm{\sum_{i\ne j}\varrho_{ijij}\ket{ij}\bra{ij}}_1\nonumber\\
		=&\norm{\sum_{i}(\lambda_i-\varrho_{iiii})\ket{ii}\bra{ii}+\sum_{i\ne j}[(\sqrt{\lambda_i\lambda_j}-\varrho_{iijj})\ket{ii}\bra{jj}]}_1+\norm{\sum_{i\ne j}\varrho_{ijij}\ket{ij}\bra{ij}}_1\nonumber\\
		=&\norm{\sum_{i}(\lambda_i-\varrho_{iiii})\ket{ii}\bra{ii}+\sum_{i\ne j}[(\sqrt{\lambda_i\lambda_j}-\varrho_{iijj})\ket{ii}\bra{jj}]}_1+\sum_{i\ne j}|\varrho_{ijij}|\nonumber\\
		\ge &\norm{\sum_{i}(\lambda_i-\varrho_{iiii})\ket{ii}\bra{ii}+\sum_{i\ne j}[(\sqrt{\lambda_i\lambda_j}-\varrho_{iijj})\ket{ii}\bra{jj}]}_1+\norm{\sum_{i< j}(\varrho_{iijj}\ket{ii}\bra{jj}+\overline{\varrho_{iijj}}\ket{jj}\bra{ii})}_1\nonumber\\
		\ge&\norm{\sum_{i}(\lambda_i-\varrho_{iiii})\ket{ii}\bra{ii}+\sum_{i\ne j}\sqrt{\lambda_i\lambda_j}\ket{ii}\bra{jj}}_1\nonumber\\
		=&\norm{\ket{\psi}\bra{\psi}-\sum_i\varrho_{iiii}\ket{ii}\bra{ii}}_1\nonumber\\
		=&\norm{\ket{\phi}\bra{\phi}-\sum_i \varrho_{iiii}\ket{i}\bra{i}}_1.
		\label{s3}
		\end{align}
	\end{widetext}
	In the first equality, as $\ket{\psi}=\sum_{i}\sqrt{\lambda_i}\ket{ii}$, the coefficients of $\ket{\psi}\bra{\psi}$ in the basis $\ket{ij}\bra{ij}$ are 0.  The first inequality is due to that the positive operator $\varrho$ satisfies the PPT condition, then $|\varrho_{iijj}|\le \sqrt{|\varrho_{ijij}\varrho_{jiji}|}\le\frac{|\varrho_{ijij}|+|\varrho_{jiji}|}{2}.$ The last inequality is due to the triangle inequality of the trace norm. In the last equality, $\ket{\phi}=\sum_i\sqrt{\lambda_i}\ket{i}.$ Next in \cite{johnston2018modified}, the authors showed that the optimal coherence states for a pure state in terms of the modified trace norm of coherence is a diagonal matrix, and in the last equality of $(\ref{s3})$, this problem is turned into the modified trace norm of coherence for a pure state, then we finish the theorem.
\end{proof}
\end{document}